\documentclass[compsoc,conference,a4paper,10pt,times]{IEEEtran}
\IEEEoverridecommandlockouts

\usepackage{cite}
\usepackage{amsmath,amssymb,amsfonts}
\usepackage{algorithmic}
\usepackage{graphicx}
\usepackage{textcomp}
\usepackage{bmpsize}
\usepackage{xcolor}
\usepackage{lipsum}
\usepackage[colorlinks=true,urlcolor=black]{hyperref}
\def\BibTeX{{\rm B\kern-.05em{\sc i\kern-.025em b}\kern-.08em
    T\kern-.1667em\lower.7ex\hbox{E}\kern-.125emX}}

\usepackage{natbib}

\usepackage[utf8]{inputenc} % allow utf-8 input
\usepackage[T1]{fontenc}    % use 8-bit T1 fonts
\usepackage{hyperref}       % hyperlinks
\usepackage{url}            % simple URL typesetting
\usepackage{booktabs}       % professional-quality tables
\usepackage{amsfonts}       % blackboard math symbols
\usepackage{nicefrac}       % compact symbols for 1/2, etc.
\usepackage{microtype}      % microtypography
\usepackage{xcolor}         % colors

\usepackage{times}
\usepackage{color}
\usepackage{amsmath}
\usepackage{graphicx,psfrag,epsf}
\usepackage{enumerate}
\usepackage{hyperref}
\usepackage{amssymb,amsfonts,mathtools}
\usepackage{bm}
\usepackage{color, colortbl}
\usepackage{algorithm}
\usepackage{algorithmic}
\usepackage{balance,cite}
\usepackage{multirow}
\usepackage{wrapfig,bbm}
\usepackage{bbm}
\usepackage{color}

\usepackage[font=small,labelfont=bf]{caption}
\usepackage{wrapfig}

\DeclarePairedDelimiterX{\norm}[1]{\lVert}{\rVert}{#1}
\DeclarePairedDelimiterX{\bnorm}[1]{\biggl\lVert}{\biggr\rVert}{#1}
\DeclarePairedDelimiterX{\abs}[1]{\lvert}{\rvert}{#1}

\newtheorem{definition}{Definition}

\newtheorem{theorem}{Theorem}

\newtheorem{proof}{Proof}

\def\de{\overset{\Delta}{=}} 
\def\R{\mathbb{R}}
\def\A{K_1}

\def\ind{\mathbbm{1}} % indicator 
\def\card{\textrm{card}}

\def\N{\mathcal{N}}
\def\D{\mathcal{D}} %datasets
\def\M{\mathcal{M}} 
\def\P{{ \mathrm{pr} }}
\def\v{\varepsilon}
\def\D{D} %database
\def\A{A} %sub-database, query
\def\F{\mathcal{F}} %query events
\def\N{N} %\N(\A): count of $\A$
\def\Nt{\hat{N}} %\Nt(\A): perturbed count of $\A$
\def\Alg{\texttt{Algo}} %  algorithm, mapping $\A \mapsto \Nt(\A)$
\def\U{U} %utility of an algorithm for a set of queries
\def\bud{\texttt{budget}}
\def\M{M} %size of a query event
\def\lap{\texttt{Lap}}

\def\q{Q}
\def\w{\omega} % elementary subdatabase
\def\enc{\textbf{1}}
\def\DD{\tilde{\D}}
\def\AA{\tilde{\A}}

\title{Private Queries with Sigma-Counting}

\author{
  Jun Gao \\
  Independent Researcher\\
  \texttt{0618johnny@gmail.com}  
  \and
  Jie Ding \\
  University of Minnesota Twin Cities \\
  \texttt{dingj@umn.edu}
}
\begin{document}

\maketitle

\begin{abstract}

% CURRENT
% \textbf{Summary}: We develop a sigma algebra-based method for generating data-private responses to counting queries with desirable utility.

% CURRENT
Many data applications involve counting queries, where a client specifies a feasible range of variables and a database returns the corresponding item counts. A program that produces the counts of different queries often risks leaking sensitive individual-level information. A popular approach to enhance data privacy is to return a noisy version of the actual count. It is typically achieved by adding independent noise to each query and then control the total privacy budget within a period. This approach may be limited in the number of queries and output accuracy in practice. Also, the returned counts do not maintain the total order for nested queries, an important feature in many applications. This work presents the design and analysis of a new method, sigma-counting, that addresses these challenges. Sigma-counting uses the notion of sigma-algebra to construct privacy-preserving counting queries. We show that the proposed concepts and methods can significantly improve output accuracy while maintaining a desired privacy level in the presence of massive queries to the same data. We also discuss how the technique can be applied to address large and time-varying datasets.

\end{abstract}


% \begin{IEEEkeywords}
% counting query, privacy, sigma-algebra 
% \end{IEEEkeywords}

% \vspace{-0.1in}
\section{Introduction}
\label{sec_intro}
% \vspace{-0.1in}

% (Motivation)
In a growing number of machine learning applications, clients will query the counts of items satisfying particular constraints from a database. Common examples are counting of user clicks in recommendation systems~\citep{ricci2011introduction}, digital marketing~\citep{krishna1995extending}, and nowcasting~\citep{youn2016nowcast}.
Specifically, in the digital marketing business, a commercial client may query the count of potential users in a particular area from a third-party database; in urban planning, a client may query the count of a cohort of residents from a government database. These counting queries are often realized by putting particular criteria on the available variables.
Examples are `individuals who live in city A,’ `users who log in more than three times a day and who are above 30,’ and `residents whose salary are more than \$5k and who own at least one car.’
%[Add an open survey result that shows the business market]

However, counting queries are calculated from cohorts of individuals, and it inevitably leaks some individual-level information. For example, whether a particular person has been to a local clinic may be sensitive information but may be revealed by an adversary, who sends a large number of queries to the clinic. Because there is a risk of leaking an individual's sensitive information, we need a privacy-preserving technique to manage the counting queries. 
%A typical risk is that given a large enough amount of queries is available to one certain employee or adversary since it is possible to recover all the partitions' population sizes. 

A popular and intuitive method to enhance data privacy is to add independent noise to each query's actual count. Depending on the type of noise to add, there are different privacy mechanisms such as the Gaussian or Laplacian mechanisms. Under a suitable data privacy framework, trade-offs between the privacy and noise level can be rigorously quantified. We will base our analysis on the framework of differential privacy (DP)~\citep{dwork2006calibrating,dwork2011differential}, a notion particularly suitable for database queries. 
It quantifies privacy leakage by a \textit{privacy parameter} $\v$ that bounds the likelihood ratio of the perturbed output under two databases differing in a single row. A smaller value of $\v$ indicates better privacy. 
We will take the Laplacian mechanism as an example. To achieve $\v$ privacy leakage, we need the noise scale to be proportional to $1/\v$.
When there are $k$ separate queries to the same database, differential privacy's composition property ensures that the overall privacy leakage is $k\v$~\citep{mcsherry2009privacy}.

For many governmental or commercial entities that provide database services, it is critical to improve the response accuracy and number of allowed queries simultaneously. However, in most practical applications, a database manager will need to pre-specify the total privacy budget (at least within a period). The budget is an upper bound of privacy leakage (e.g., the above $k\v$). Once the privacy leakage goes beyond the budget, the database will no longer respond. This directly leads to trade-offs between the number of allowed queries and the accuracy of each response. 

In practice, however, 
% this approach is limited by the utilitarian mechanism within a fixed privacy budget and loss in output accuracy. 
there is usually a massive number of queries to the same database. To maintain a fixed privacy budget, the database manager needs to impose a considerable amount of noise on each query response, which hinders the output accuracy.
Moreover, another critical yet often overlooked issue is that the returned counts do not maintain the total order for nested queries. The total order means that if two queries, corresponding to sub-databases $A_1$ and $A_2$, satisfy $A_1 \subset A_2$, it holds that $\N(A_1) \leq \N(A_2)$. Here $\N(\cdot)$ denotes the returned count. Maintaining total order is crucial in many real-life applications, e.g., data publishing by the U.S. Bureau of Labor Statistics~\citep{groshen2015opportunities}, recommendation systems~\citep{ricci2011introduction}, and various downstream machine learning based on ordinal data~\citep{oh2015collaboratively}.

%% CONTRIBUTIONS
% \vspace{-0.1in}
\subsection{Contributions} \label{subsec_cont}
% \vspace{-0.1in}
% This main contributions of this work are three folds.

% $\bullet$ 
We identify two critical challenges of generating privacy-preserving counting responses, namely the `budget excess’ and `non-monotonic counts’. 
% $\bullet$ 
Then, we present the design and analysis of a new method, sigma-counting, that addresses the challenges above. Sigma-counting uses the notion of sigma-algebra to construct privacy-preserving counting queries to a database. It enjoys the following properties: 1) High utility in the presence of a large number of queries, and 2) Monotonicity, meaning that the returned count of a subset $\A_2$ will not be smaller than that of a larger subset $\A_1$ ($\A_1 \subseteq \A_2$). 

% To the best of our knowledge, this is the first work that addresses the above important properties in privacy-preserving counting queries. 
Sigma-counting fully utilizes the nature of database structure: many queries are overlapping, and many sub-databases are repetitively included in the queries. We show that the proposed novel concepts and methods can significantly improve output accuracy while achieving the same privacy level. On the other hand, given the same privacy budget and output accuracy, sigma-counting can respond to many more queries than existing methods.
% The main contributions are supported by new concepts, methods, and theoretical analysis.
% The proposed method is shown to significantly outperform the existing popular approach in terms of utility-privacy tradeoffs. 
We also extend the proposed method to handle databases with a large number of attributes or time-varying entries.  
\subsection{Related work}
% \vspace{-0.1in}
%We propose a novel privacy-preserving method for responding to counting queries. 

%\textit{Data privacy}.
A popular framework of evaluating data privacy is through differential privacy (DP)~\citep{dwork2006calibrating,dwork2011differential}, a cryptographically motivated definition of privacy that has been studied over the past decade across different fields~\citep{chaudhuri2011differentially,sarwate2013signal,dong2019gaussian,neunhoeffer2020private,vietri2020new}. DP measures privacy leakage by bounding the likelihood ratio of the output of an algorithm under two databases differing in any single individual. 
DP is statistical database privacy, which is particularly suitable for privatizing individual identity in a database. 
Unlike database privacy, there is the notion of local privacy that aims to protect privacy during the data collection stage, without trusted database management. Examples are local differential privacy~\citep{evfimievski2003limiting,dwork2006calibrating,kasiviswanathan2011can}, secure multi-party computing~\citep{yao1982protocols,chaum1988multiparty} and homomorphic encryption~\citep{gentry2009fully,armknecht2015guide}. % TO EDIT AFTER ACCEPTANCE

Our work focuses on the database privacy scenario. It involves three parties: data owner (often an individual) who provides sensitive data, database (a trusted third-party) who manages data collected from their owners), and client (often data scientists) who queries statistics from the database. 
There has been an active line of research in database queries.
\citet{batchlinearqueries} developed a method to address batches of linear counting queries, which used low-rank mechanisms to efficiently apply DP mechanism and achieve high accuracy.
\citet{dphistogram} developed a method to construct the optimal DP-compliant histograms from the original count sequences.
\citet{privatesql} developed a two-step method to integrate DP into database querying.
\citet{johnson2018towards} developed practical techniques for DP in real-world SQL-based analytical systems.
\citet{li2010optimizing} developed a method named matrix mechanism to respond to a workload of predicate counting queries.
\citet{proserpio2012calibrating} showed that scaling down the contributions of challenging records non-uniformly can improve accuracy by bypassing the worst-case requirements of noise magnitudes.
\citet{mcsherry2009privacy} developed a programming interface for DP-based data analysis through a SQL-like language.
\citet{mohan2012gupt} developed a model of data sensitivity that degrades data privacy over time to enable efficient allocation of different levels of privacy, while guaranteeing an overall constant privacy level.
Different from the many existing works that have addressed system development and privacy-preservation for particular queries, we focus on the challenges regarding massive overlapping queries to the same database (Subsection~\ref{subsec_cont}). 

The remaining paper is organized as follows.
In Section~\ref{sec_form}, we describe the problem formulation.
In Section~\ref{sec_3}, we describe the main idea of sigma-counting, its general workflow, and detailed algorithms. 
In Section~\ref{sec_con}, we make our conclusion. 
In the supplementary document, we provide some additional proofs and experimental studies on both simulated and real-world data.

% \vspace{-0.1in}
\section{Problem Formulation} \label{sec_form}
% \vspace{-0.1in}

\subsection{Background: database query}
% \vspace{-0.1in}

We first introduce some notations about database queries.

% \vspace{-0.1in}
\begin{definition}[Database]
A database $\D$ is a $n \times p$ matrix. 
Each row representing a unique object, and each column represents a feature. %, which describes a trait of each person. 
\end{definition}
% \vspace{-0.1in}

For example, a database may consist of a set of mobile users and their log-activity-related features or a set of residents in a city with demographic features. 

% \vspace{-0.1in}
\begin{definition}[Counting query]
A counting query is a constraint on $p$ features sent from a client to the database $\D$, and a response is the number of rows in $\D$ satisfying the constraint. 
\end{definition}
% \vspace{-0.1in}

Upon the receipt of a query, the database can filter out an $n' \times p$ ($n' \leq n$) submatrix of $\D$, and return the number of rows $n'$ to the client. 
In many practical applications, e.g., the Machine-Learning-as-a-Service (MLaaS)~\citep{alabbadi2011mobile, ribeiro2015mlaas} and digital marketing~\citep{krishna1995extending,ricci2011introduction}, the client is possibly an adversary that aims to reverse-engineer sensitive individual information from the query responses. An emerging concern of this kind is the membership inference attacks~\citep{shokri2017membership}, where an adversarial client aims to infer whether a particular individual  (a row of features) exists in the database or not.
A consequence of the attack is that an adversary can reconstruct the database.
A principled way to enhance privacy is to add noises to the actual count and evaluate trade-offs under a data privacy framework.

% \vspace{-0.1in}
\subsection{Background: differential privacy (DP)}
% \vspace{-0.1in}

Differential privacy~\citep{dwork2004privacy,dwork2006our,dwork2011differential} is a framework that evaluates the worst-case scenario of leaking the existence of a single individual from potential adversaries. It is particularly suitable for database privacy, where a client queries summary statistics (such as counts) from the database without accessing the raw data. 

% \begin{definition}
% 	Let $\Y$ denote a set of private data, and $\Y_{-y}$ the set that excludes an element $y\in\Y$. Let $Z$ denote a random variable with a conditional distribution density on $\Y$ or $\Y_{-y}$.
%  For a given privacy parameter $\alpha \geq 0$, a privacy mechanism $\M$ is $\alpha$-differentially private if for all $y\in \Y$,
% 	\begin{align}
% 	   e^{\alpha} \leq \frac{p(z \mid \Y)}{p(z  \mid Y_{-y})} \leq e^{\alpha} .\label{eq31}	
% 	\end{align}
% \end{definition}

% \vspace{-0.1in}
\begin{definition}[$\v$-differential privacy]\label{eq_DP}
A randomized mechanism $\M: \D \mapsto \mathcal{R}$ with domain $\D$ and range $\mathcal{R}$ satisfies $\v$-differential privacy (DP) if for any two adjacent inputs $\D_1,\D_2$ and for any subset of outputs $S \subseteq \mathcal{R}$, it holds that $\P(\M(\D_1) \in S) \leq e^{\v} \P(\M(\D_2) \in S)$. %$ + \delta$.
\end{definition}
% \vspace{-0.1in}

Here, the adjacency means that $\D_1,\D_2$ only differ in one row. 
Briefly speaking, differential privacy measures privacy leakage by a \textit{privacy parameter} $\v$ that bounds the likelihood ratio of the output of a (private) algorithm under two databases differing in a single individual. The smaller $\v$ is, the more privacy for any single individual. 
The original definition of $\v$-differential privacy was also extended to allow an additive term $\delta$, so that $\P(\M(d) \in S) \leq e^{\v} \P(\M(d') \in S)+\delta$~\citep{dwork2006our}. % We use the variant introduced by \citep{dwork2006our}, which allows for the possibility that plain $\v$-differential privacy is broken with probability $\delta$. 

A standard randomized mechanism is described below. Suppose that a query $f: \D \rightarrow \R^d$ satisfies 
\begin{align}
\norm{f(\D_1)-f(\D_2)}_1 \leq c_f \label{eq_3}
\end{align}
for any adjacent $\D_1,\D_2$ and some constant $c_f>0$. Here, $\norm{\cdot}_1$ is the $L_1$ norm.
Let $\texttt{Lap}(\lambda)$ denote the Laplace distribution with density function $p_{\lambda}(x)=(2 \lambda)^{-1} \exp(-|x|/\lambda)$.
Then, it can be verified according to Definition~\ref{eq_DP} that
\begin{align}
    \M(\D) = f(\D) + \zeta, %\quad \zeta \sim \texttt{Lap}(c_f/\v)
    \label{eq_mech}
\end{align}
where $\zeta$ consists of i.i.d. $\texttt{Lap}(c_f/\v)$,
is a mechanism that satisfies $\v$-differential privacy. The intuition is that more privacy (namely smaller $\v$) requires larger noise. 
An important property of differential privacy is the {sequential composition}. if there are $n$ independent mechanisms: $\M_1,\ldots,\M_k$, whose privacy parameter are $\v_1,\ldots,\v_k$ differential privacy, respectively, then any composition $g$ of them, $g(\M_1,\ldots,\M_k)$, is $\sum_{i=1}^{k}\v_i$-differentially private~\citep{mcsherry2009privacy}.
This property says that if each module in a system accesses data in a differentially private way, their modular combination is also private with a quantifiable privacy parameter. It leads to the concept of the privacy loss budget.

% \vspace{-0.1in}
\begin{definition}[Privacy budget]
    A privacy budget (denoted by $\v_{\bud}$) is the largest differential privacy parameter allowed for a database to respond to queries.
\end{definition}
% \vspace{-0.1in}

For example, a database manager will control the total budget in a particular period. Suppose that it has made $m$ differentially-private responses, each with privacy parameter $\v$. Then it will stop at the $m$th query with  $m=\lfloor \v_{\bud}/\v \rfloor$.
% with privacy parameters $\v_1,\v_2,\ldots$. Then the it will stop at the $m$th query with   
% \begin{align}
%     \v_1+\cdots+\v_m \leq \v_{\bud}, \label{budget}
% \end{align} 
% and $\v_1+\cdots+\v_{m+1} > \v_{\bud}$.
% It is simplified ... 
% The database manager need not restrict particular types of queries. 

%  	\begin{algorithm*}[tb]
% 		\centering
% 		\caption{Generic Algo}\label{algo}
% 		\footnotesize
% 		\begin{algorithmic}[1]
% 			\renewcommand{\algorithmicrequire}{\textbf{Input:}}
% 			\renewcommand{\algorithmicensure}{\textbf{Output:}}
% 			\REQUIRE A
% 			\\ \hrulefill Stage I: Learning  \hrulefill \\
% 			\STATE A fits a model using $(y_a, X_a)$ as labeled data, by 
% 			\ENSURE A obtains 
% 		\end{algorithmic}
% 	\end{algorithm*}

% \vspace{-0.1in}
\subsection{Benchmark algorithm} \label{subsec_bench}
% \vspace{-0.1in}

In the context of a counting query, a client sends a query that can be associated with a unique sub-database $A$, and the database manager returns the number of rows in $\A$, denoted by $\N(A)$. Following the notation in (\ref{eq_3}), the counting query function is $f: A \mapsto N(A)$. Clearly, any counting query function $f$ has a sensitivity of $c_f = 1$. 
A standard method used in the literature is to respond the count using the mechanism in (\ref{eq_mech})~\citep{mcsherry2009privacy,proserpio2012calibrating,mohan2012gupt,li2010optimizing}).
In other words, the returned count is
\begin{align}
    \Nt(A) = \N(A) + \zeta, \quad \zeta \sim \lap(1/\v) \label{eq_4}
\end{align}
to ensure $\v$-differential privacy in one counting query, and continue responding until a pre-specified budget is reached.
We call this method the benchmark method, summarized in Algo.~\ref{algo_tradition}. We point out that the Laplacian noise can be replaced with other noise distributions such as  Gaussian. For notational convenience, we demonstrated the methods using Laplacian noise.
One may also round $\Nt(A)$ to its nearest integer to enhance interpretability in practice. It will not increase the privacy leakage due to DP's robustness to post-processing~\citep{dwork2011differential}.
This benchmark algorithm has time complexity $O(1)$ and space complexity $O(1).$ 

% \begin{wrapfigure}{R}{0.47\textwidth}
% \vspace{-8mm}
% \begin{minipage}{0.47\textwidth}
 	\begin{algorithm}[H]
		\centering
		\caption{Benchmark algorithm $\Alg_{1}$}\label{algo_tradition}
		\footnotesize
		\begin{algorithmic}[1]
			\renewcommand{\algorithmicrequire}{\textbf{Input:}}
			\renewcommand{\algorithmicensure}{\textbf{Output:}}
			\REQUIRE Pre-specified budget $\v_{\bud}$, a set of queries represented by $A_1,A_2,\ldots$
% 			\\ \hrulefill Stage I: Learning  \hrulefill \\
            \FOR{$t=1 \to \q \de \lfloor \v_{\bud}/\v \rfloor$}
			\STATE Respond the $\Nt(A_t)$ defined in (\ref{eq_4}) 
			\ENDFOR
			\ENSURE The returned queries $\Nt(A_1),\ldots,\Nt(A_{\q})$
		\end{algorithmic}
	\end{algorithm}
% \end{minipage}
% \end{wrapfigure}
%

In the benchmark algorithm, the total budget is often pre-fixed according to regulations. It clearly shows trade-offs between the number of responses and the accuracy of each response. In particular, for an enormous number of responses $\q$, the value of $\v \leq \v_{\bud}/\q$ has to be small, indicating a considerable noise added to each count response. We will introduce a new algorithm which significantly improves accuracy within the same privacy budget. 
For notational convenience, we introduce the following notions.

% \vspace{-0.1in}
\begin{definition}[Counting-query algorithm] \label{def_algo}
    A counting-query algorithm is defined by a set of mappings $\A \mapsto \Nt(\A)$ for all $A \subset \D$.
\end{definition}
% \vspace{-0.1in}

The benchmark Algo.~\ref{algo_tradition} is a counting-query algorithm. 
We will quantify any algorithm's utility using the average accuracy as evaluated by the average variance of $\N(A)$ conditional on the actual counts.

% \vspace{-0.1in}
\begin{definition}[Utility]
\label{utility}
The utility of a privacy algorithm for a finite set of queries $\F=\{\A\}$ is 
$$
% \U(\F, \Alg)
\U = \frac{1}{\card(\F)} \sum_{\A \in \F} \frac{2}{Var(\hat{N}(A) \mid N(A))}.
$$
\end{definition}
% \vspace{-0.1in}

\begin{figure}[tb]

%  \begin{minipage}{\textwidth}
  \begin{minipage}[b]{0.52\textwidth}
    % \begin{figure}[tb]
	\centering
	\includegraphics[width=1\linewidth]{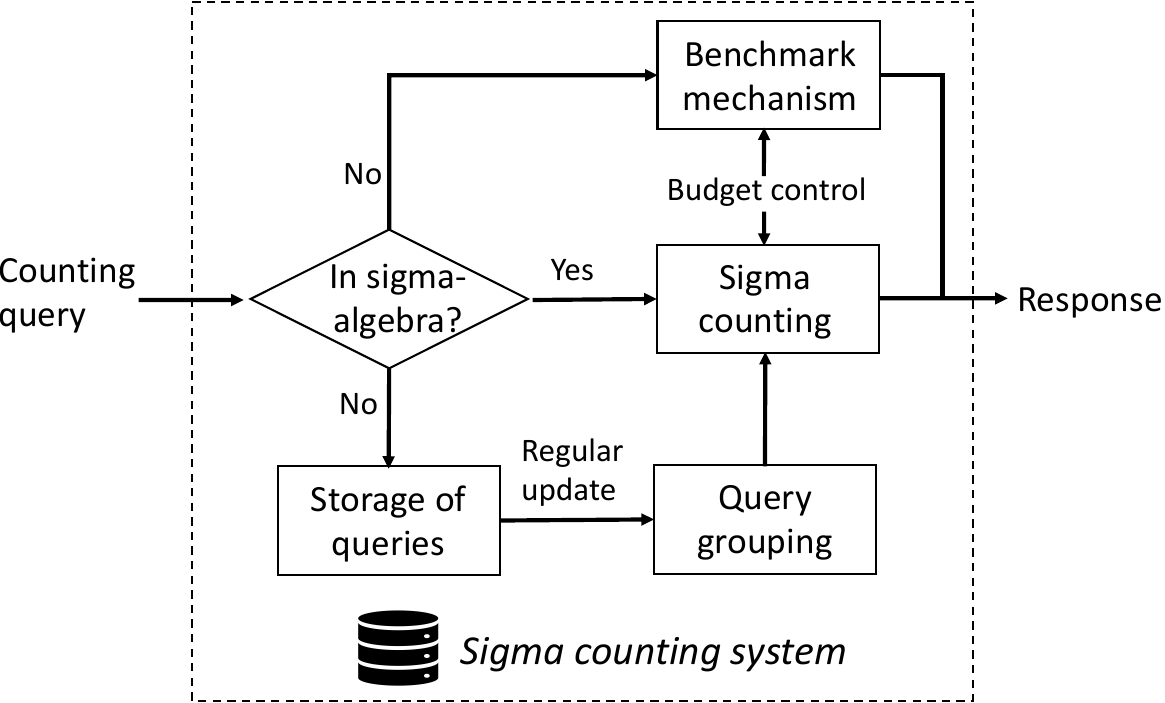}
    % \vspace{-0.3in}
	\captionof{figure}{Overview of the sigma-counting system.} 
	\label{fig3}
    % \end{figure}
%   \end{minipage}
%   \hfill
% \begin{table}[tb]

% \begin{minipage}[b]{0.4\textwidth}
\centering
{\small
\captionof{table}{A summary of notations}
\scalebox{0.8}{
\begin{tabular}{ll}
\toprule
Notation                               & Name                         \\\midrule
$\D$                                   & database, set of rows        \\
$\Omega$                               & set of sub-databases         \\
$|\Omega|$                             & total number of sub-databases         \\
$\omega$                               & element of $\Omega$          \\
$\A$                                   & subset of $\Omega$, query    \\
$\F$                                   & set of subsets in $\Omega$                \\
$\M(\A)$                               & size of $\A$, or the number of $\omega$'s in $\A$                 \\
$\N(\A)$                               & count in $\A$, total number of rows                \\
$\Nt(\A)$                              & responded/perturbed count in $\A$       \\
$\Alg$                                 & algorithm, mapping $\A \mapsto \Nt(\A)$ for a set of $\A$'s \\
$\U$                                   & utility, for an algorithm and a set of $\A$'s     \\
$\v$                                   & DP parameter \\ %random noise                 \\
% $\E$                                   & expectation                  \\
% $\p$                                   & probability distribution     \\
$\q$                                   & total number of queries     \\
\bottomrule
\end{tabular}
}
}
% \end{table}
%   \end{minipage}
  \end{minipage}
  
\end{figure}

% \vspace{-0.1in}
\section{Sigma-Counting System} \label{sec_3}
We first describe the main idea of sigma-counting. Suppose that there is a large number of overlapping queries to the same database. In the extreme case, there are $2^n$ possible queries, where $n$ is the number of rows. If the database manager applies Algo.~\ref{algo_tradition} to each of them, it requires that $\v \leq \v_{\bud}/2^n$, corresponding to a low response accuracy on average. Alternatively, we can apply noise to each row and then report the sum of each perturbed count.
For example, suppose that a database contains $n$ rows, where the $i$th row indicates the number of local residents in the range of $[a_i,a_{i+1})$, noted by $n_i$. Then, we generate $\hat{n}_i = n_i+\zeta_i$, where $\zeta_i$ denotes independent noise. For any future query in the range of $[a_k, a_j)$, the manager will return $\sum_{i=k}^{j-1} \hat{n}_i$.

We explain the general workflow of a sigma-counting system (illustrated in Fig.~\ref{fig3}). 

\textbf{Clustering queries}.
Suppose that there is a set of pre-existing queries known to the database manager.
The manager will first perform clustering to vertically-split the database variables into disjoint sets so that features across sets are less jointly queried than those within the same set. This module will simplify computation and enhance privacy-utility trade-offs in sigma-counting.

\textbf{Sigma-counting}. 
The manager will partition the database into a set of elementary sub-databases so that each query corresponds to a union of some of the elementary sub-databases. The notion of elementary sub-databases and queries are reminiscent of the outcome and events in probability, and they can be cast into one sigma-algebra. The database manager will then apply the mechanism (\ref{eq_4}) to the count of each sub-database. For each query, it will  add the perturbed counts from the corresponding sub-databases. The provider will separately apply the mechanism (\ref{eq_4}) to each query that involves cross-cluster variables. 

\textbf{Online queries}. 
Upon each new online query, the manager will check if it belongs to the existing sigma-algebra. If so, it will return the count from sigma-counting; Otherwise, it will separately apply benchmark algorithm~\ref{eq_4} and return it to the client. The manager will regularly perform clustering from historical queries.

\textbf{Evolving database}.
Upon the change of a row in the database, the manager will do the following.
For all the possible sub-databases affected by this change, their perturbed counts will change in the same way as actual counts change. The privacy leakage will be correspondingly adjusted. %So, the randomness of sub-databases remain the same. After this is completed, 
Any incoming online query will be responded in the same manner as described in the previous paragraph, until a pre-specified privacy budge is reached.

% In more complex application scenarios, the database $\D$ being queried can be time-varying. For example, a cohort of individuals occasionally update their demographic information in the database. Responding to the same query sent to a time-varying database at different time steps is more likely to leak individual information compared with that to a static database. We denote the database as $\D_{1:T} \de \{D_t, t= 1,\ldots,T\}$, where $t$ represents the time step and $D_t$ denotes the database at time $t$. We first introduce  differential privacy for a time-varying database.

We will introduce the key concepts, algorithms, and analysis of each module in the following subsections.
% We briefly summarize the fundamental ideas of each module of the system.

\subsection{Sigma-counting method} \label{subsec_offline}
% \vspace{-0.1in}

% \textbf{Offline Sigma-counting}. 

% \begin{wrapfigure}{R}{0.47\textwidth}
% \vspace{-8mm}
% \begin{minipage}{0.47\textwidth}
 	\begin{algorithm}[tb]
		\centering
		\caption{Sigma-counting ($\Alg_{2}$)}\label{algo_offline}
		\footnotesize
		\begin{algorithmic}[1]
			\renewcommand{\algorithmicrequire}{\textbf{Input:}}
			\renewcommand{\algorithmicensure}{\textbf{Output:}}
			%\REQUIRE The database manager pre-specifies a Sigma Algebra $\Omega$, and assume all the queries are an element of $\Omega.$
			\REQUIRE A set of queries represented by $\A_1,\ldots,A_{\q}$. 
% 			\\ \hrulefill Stage I: Learning  \hrulefill \\
            \STATE Induce a sigma-algebra $(\Omega, \F)$  from the queries
            \STATE For each element/sub-database $\w \in \Omega$, apply the perturbation scheme (\ref{eq_4}) to the query $\{\w\}$
            \FOR{any query $\A$}
			\STATE Respond
			\begin{align}
            \Nt(\A) = \sum_{\w \in \A} \Nt(\{\omega\}) \label{sig}
            \end{align}
			\ENDFOR
			\ENSURE The returned counts $\Nt(A_1),\ldots,\Nt(A_{\q})$
		\end{algorithmic}
	\end{algorithm}
% \end{minipage}
% \end{wrapfigure}
%
Suppose that the underlying database $\D$ is partitioned into a set of sub-databases of $D$, denoted by $\w_1,\ldots,\w_k$. In other words, $\w_i \cap \w_j = \emptyset$ if $i \neq j$, and $\cup_{i=1}^k \w_i = D$.
% With slight abuse of notation, we use $D_1 \cup D_2$ to denote the larger database concatenated from $D_1$ and $D_2$.
We treat each $\w_1$ as an `element' of the set $\Omega \de \{\w_1,\ldots,\w_k\}.$

Recall that a sigma-algebra on a set $\Omega$ is a collection $\F$ of subsets of $\Omega$ that includes $\Omega$ itself, is closed under complement, and is closed under countable unions.
We use $(\Omega, \F)$ to denote the largest sigma-algebra generated by the finite set $\Omega$. In other words, $\F$ is the set of possible subsets of $\Omega$. Any element of $\F$, denoted by $\A \subset \Omega$, is interpreted as the sub-database associated with a particular query, or simply a `query.'
For example, if $\A=\{\w_1, \w_2\}$, it contains the data items in $\w_1$ or $\w_2$. With a slight abuse of notation, we also write $\A=\w_1 \cup \w_2$ and treat it as a sub-database of $\D$. We sometimes write $\omega$ and $\{\omega\}$ interchangeably. 

In practice, the database will receive a set of frequently seen queries. Suppose that $\A_1,\ldots,\A_t$ are those pre-specified queries. They naturally induce a sigma-algebra $(\Omega, \F)$  that is the smallest among all those that generate the queries. In other words, $\Omega$ is the smallest set that contains each $\A_i$ as its subset.
Following the above notation, we introduce the following counting-query algorithm (Definition~\ref{def_algo}).

For each $\w\in\Omega$, we apply the following mechanism to obtain the count. %$\Nt(A)$, 
\begin{align}
    \Nt(\omega) = (\N(\omega) + \zeta)^{+}, \quad \zeta \sim \lap(1/\v) \label{eq_10}
\end{align}
Here, $x^{+} \de \max(x,0)$.
For a general query $\A \in \F$, we then calculate the count to be the sum of the counts of its elements. The algorithm pseudocode is summarized in Algo.~\ref{algo_offline}.

The mechanism in (\ref{eq_10}) is slightly different with the benchmark method in (\ref{eq_4}) so that $\Nt(\omega) \geq 0.$
It will cause some bias to the response, but it is negligible for a moderately large sub-database $\omega$ in reality.

% \begin{figure}[tb]
%  \begin{minipage}{\textwidth}
%   \begin{minipage}[b]{0.49\textwidth}
%     % \begin{figure}[tb]
% 	\centering
% 	\includegraphics[width=1\linewidth]{figures/fig1.pdf}
%     % \vspace{-0.3in}
% 	\captionof{figure}{Illustration of the sigma-counting.} 
% 	\label{fig1}
% 	\end{minipage}
% 	\hfill
%   \begin{minipage}[b]{0.49\textwidth}
% 	\centering
% 	\includegraphics[width=1\linewidth]{figures/fig2.pdf}
%     % \vspace{-0.3in}
% 	\captionof{figure}{Illustration of the query grouping.} 
% 	\label{fig2}
%   \end{minipage}
%   \end{minipage}	
% \end{figure} 

\begin{figure}[tb]
%  \begin{minipage}{\textwidth}
  \begin{minipage}[b]{0.49\textwidth}
    % \begin{figure}[tb]
	\centering
	\includegraphics[width=1\linewidth]{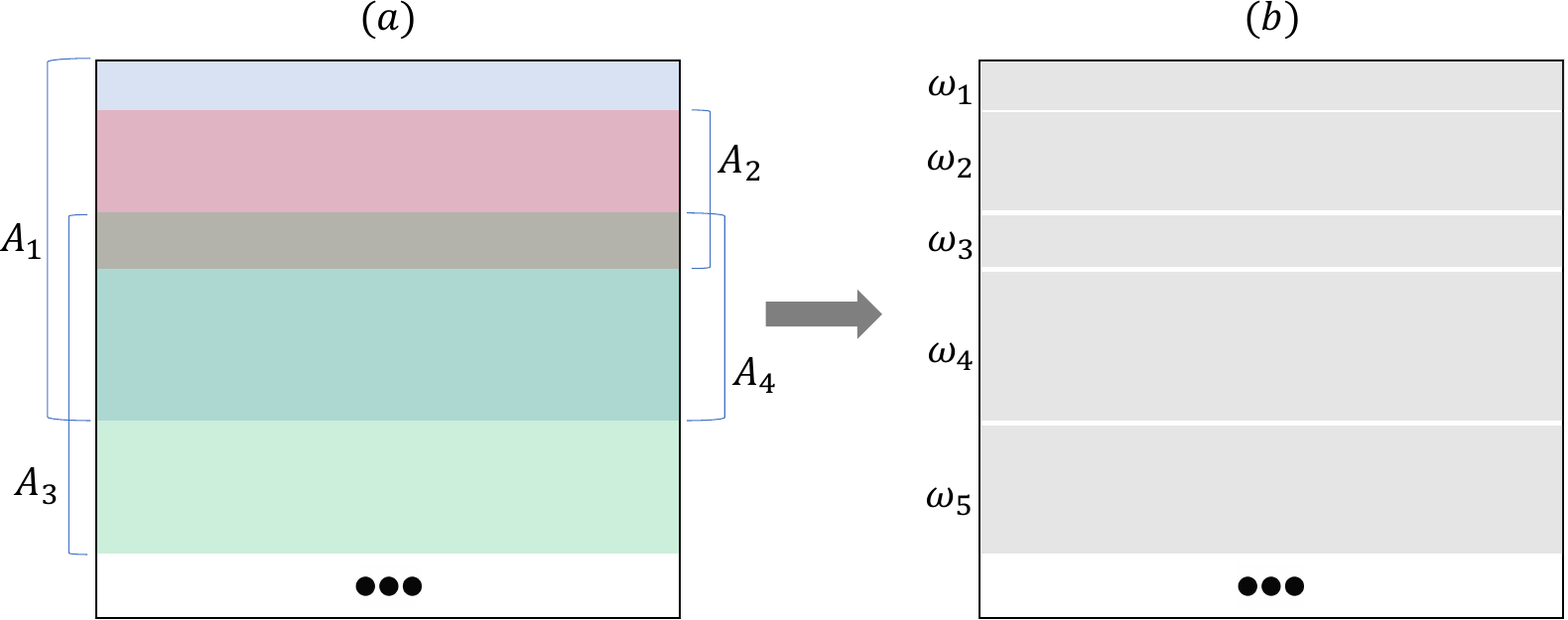}
    % \vspace{-0.3in}
	\captionof{figure}{Illustration of the sigma-counting.} 
	\label{fig1}
	\end{minipage}
% 	\hfill
%   \begin{minipage}[b]{0.49\textwidth}
% 	\centering
% 	\includegraphics[width=1\linewidth]{figures/fig2.pdf}
%     % \vspace{-0.3in}
% 	\captionof{figure}{Illustration of the query grouping.} 
% 	\label{fig2}
%   \end{minipage}
%   \end{minipage}	
\end{figure}

\begin{figure}[tb]
%  \begin{minipage}{\textwidth}
%   \begin{minipage}[b]{0.49\textwidth}
%     % \begin{figure}[tb]
% 	\centering
% 	\includegraphics[width=1\linewidth]{figures/fig1.pdf}
%     % \vspace{-0.3in}
% 	\captionof{figure}{Illustration of the sigma-counting.} 
% 	\label{fig1}
% 	\end{minipage}
% 	\hfill
  \begin{minipage}[b]{0.49\textwidth}
	\centering
	\includegraphics[width=1\linewidth]{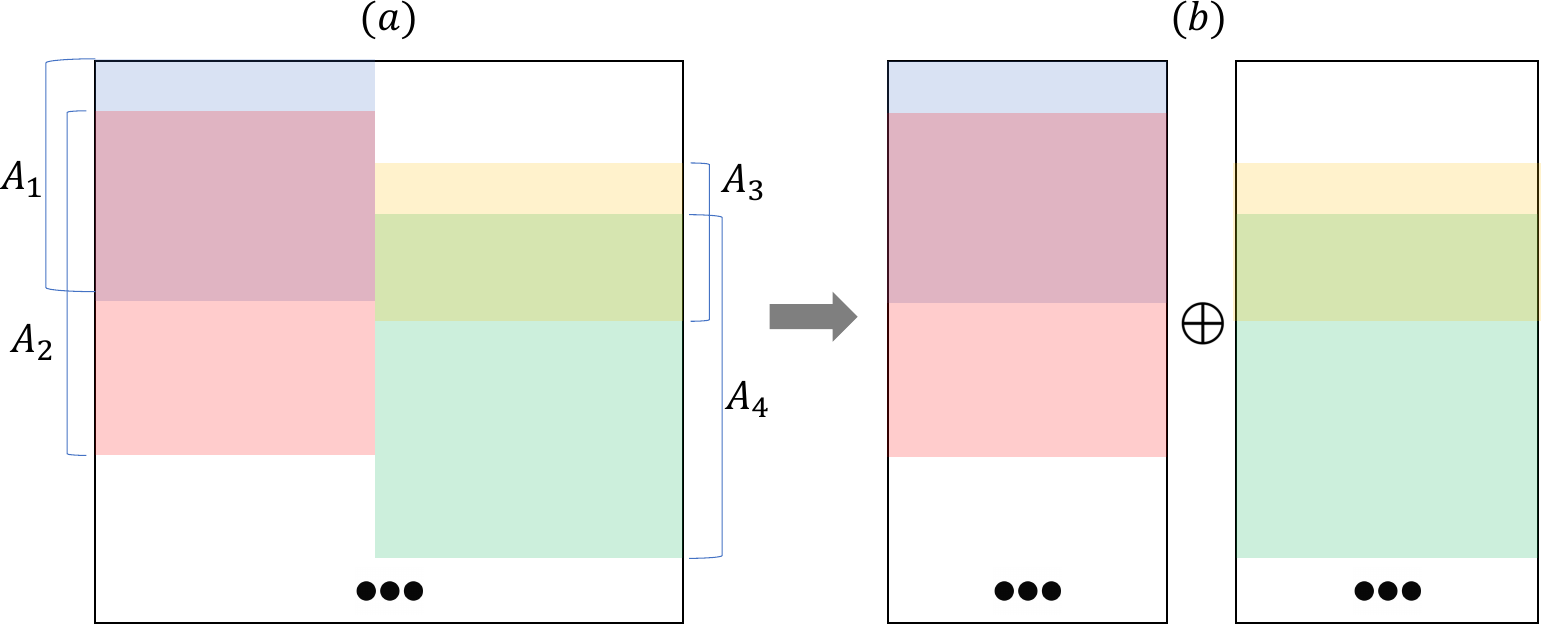}
    % \vspace{-0.3in}
	\captionof{figure}{Illustration of the query grouping.} 
	\label{fig2}
   \end{minipage}
%   \end{minipage}	
\end{figure}

For each query, this algorithm has time complexity $O(|\Omega|)$ and space complexity $O(|\Omega|)$.
The following Theorem~\ref{thm_mono} states that the proposed algorithm maintains the ordinal structure of queries. 

\begin{theorem} \label{thm_mono}
     Algorithm~\ref{algo_offline} is monotonic in the sense that for any two queries $\A_1,\A_2$ satisfying $\A_1 \subseteq \A_2$, we have $\Nt(\A_1) \leq \Nt(\A_2)$.
\end{theorem}

\begin{proof}
From (\ref{sig}) and (\ref{eq_10}), we have 
\begin{align}
    \Nt(\A_2) = \Nt(\A_1) + \sum_{\w \in \A_2-\A_1} \Nt(\w)
    \geq  \Nt(\A_1) .
\end{align}

\end{proof}

The following result shows that if the number of queries is large, the proposed method will enjoy better utility than the benchmark method.
We use $|\Omega|$ to denote the cardinality of $\Omega$, namely the number of sub-databases.
\begin{theorem}\label{thm_offline}
    Suppose that $\q$ is the total number of queries. If $\q > |\Omega|^{3/2},$ Algorithm~\ref{algo_offline} has better utility than Algorithm~\ref{algo_tradition} given the same privacy budget.
\end{theorem}

\begin{proof}
We will compute the overall privacy leakage and utility for each algorithm separately. 

% In terms of privacy budget: 
We will first analyze benchmark Algo.~\ref{algo_offline}. By the composition property~\citep{mcsherry2009privacy}, the overall privacy leakage after $\q$ queries is
\begin{align}
\v_\bud(\Alg_{\ref{algo_tradition}}) \leq \q \cdot \v.\label{eq_12}
\end{align}
Here, $\v_\bud(\cdot)$ denotes the privacy budget consumed. 
Recall our definition of utility (Definition~\ref{utility}). Since the noise variance added to each query is $2\v^{-2}$, the utility is
$ \U(\Alg_{\ref{algo_tradition}}) = \v^{2}.$

% Notice inequality comes from the capping:
% $$Var( max( N + \zeta, 0)  )) \geq  \frac{1}{4} Var( N + \zeta  )),$$
% since $N \geq 0.$

In Algo.~\ref{algo_offline}, suppose that we apply $\v^{\prime}$-DP to each $\w$ (\ref{eq_10}).
The overall privacy leakage is 
\begin{align}
\v_\bud( \Alg_{2}) = |\Omega| \cdot \v^\prime. \label{eq_11}
\end{align}
For the utility, it is lower bounded by 
\begin{align}
\U(\Alg_{2})
\geq |\Omega|^{-1} {\v^{\prime}}^{2}. \label{eq_14}
\end{align}
% ??????? should explain more
% The last $\geq$ is due to the fact that $M(A) \leq |\Omega|$ for any query A. 
Given the same privacy budget, from (\ref{eq_12}) and (\ref{eq_11}), we have 
$ %$\begin{align} 
   \v^{\prime} =  \v \cdot \q / |\Omega|. 
$ %\end{align}
Taking the above $\v^{\prime}$ into (\ref{eq_11}), we obtain 
$ %\begin{align*}
    \v_{\bud}(\Alg_{2}) = \q \cdot \epsilon,
$ %\end{align*}
which is the same as (\ref{eq_12}). Taking $\v^{\prime}$ into (\ref{eq_14}), we obtain 
$ %$\begin{align}
    \U(\Alg_{2}) \geq   \v^{2} \cdot Q^2 / |\Omega|^3.
$ %\end{align}
%
% Sigma counting has utility $\geq \v^{2} Q^2/|\Omega|^3 $.
Consequently, the condition that $
\v^{2} Q^2/|\Omega|^3 > \v^2
$ or 
$
\q > |\Omega|^{3/2}
$
will guarantee $\U(\Alg_{2})  > \U(\Alg_{1})$, which completes the proof. 
 
\end{proof}
 
We note that there are $2^{|\Omega|}$ different possible queries within the sigma-algebra. Our method can handle many repetitive queries without compromising the overall privacy budget.

\subsection{Grouped sigma-counting}
% \vspace{-0.1in}

In practice, the database often contains numerous feature variables (columns). Queries are usually not uniformly distributed. For example, the queries sent from a client are either on demographic information, such as gender and age, or on economic data, such as income and location. Suppose that we can identify the queries into different groups, so that queries within the same group are much more frequent than those across different groups. We can then further increase the utility and reduce the computation cost of the sigma-counting method in Subsection~\ref{subsec_offline}. This idea is illustrated in Fig.~\ref{fig2}.
% In terms of real situations, common queries are often 'shallow': meaning a query $A$ often does not have so many querying conditions. It makes clustering possible: clustering the queries which are asking the same categories of information can reduce complexity significantly. 

Following the above idea, we propose the following method.
Recall that a query $\A$ can be regarded as a sub-database that meets certain constraints on the $p$ variables. 
We denote $a(\A)$ as the set of variables that actively constrain the $\A$. %\subseteq [1:p]
For example, suppose that the query is on the count of $\A$=`users between 60 and 70 years old who live in California.' Then, $a(\A)$ contains the age and location variables.
We represent $\A$ as the following binary vector $\enc_{\A} \in \{0,1\}^p$, 
\begin{align}
    \enc_{\A}[i] = \ind_{i \in a(\A)}. \label{eq_nume}
\end{align}
For a set of queries $\F=\{\A_1,\ldots,\A_t\}$, we cluster them according to their numeric representations $\enc(\A_1),\ldots,\enc(\A_t)$, create sigma-algebra for each cluster, and return perturbed counting.
The pseudocode is summarized in Algo.~\ref{algo_group}.

In terms of complexity, Algo.~\ref{algo_group} has the spatial complexity $O(\sum_{k=1}^{K} |\Omega_k|)$, and time complexity $O ( K + max_{1\leq k \leq K} |\Omega_k|)$.

% \begin{wrapfigure}{R}{0.5\textwidth}
% \vspace{-5mm}
% \begin{minipage}{0.5\textwidth}
 	\begin{algorithm}[H]
		\centering
		\caption{Grouped sigma-counting ($\Alg_{\texttt{3}}$)}\label{algo_group}
		\footnotesize
		\begin{algorithmic}[1]
			\renewcommand{\algorithmicrequire}{\textbf{Input:}}
			\renewcommand{\algorithmicensure}{\textbf{Output:}}
			\REQUIRE A set of queries $\F\de \{\A_1,\ldots,\A_{\q}\}$
			\STATE Apply clustering to their numerical embeddings $\{\ind_{i \in a(\A)}: \A \in \F\}$, as introduced in (\ref{eq_nume})
			\STATE Obtain $K$ clusters of queries
			\STATE For the $k$th cluster, let the queries induce a sigma-algebra $(\Omega_k, \Sigma_k)$, $1\leq k \leq K$
            \FOR{any query $\A \in \F$}
			\STATE Suppose that $\A$ belongs the cluster $k$. Apply Algo.~\ref{algo_offline} with $(\Omega_k, \Sigma_k)$ to obtain the $\Nt(A)$
			\ENDFOR
			\ENSURE The returned counts $\Nt(\A), \forall \A \in \F$
		\end{algorithmic}
	\end{algorithm}
	%
% \end{minipage}
% \end{wrapfigure}

% This complexity can be much less than that of Algo.~\ref{algo_offline}. For example, in the extreme case where the queried variables fall into $K$ disjoint subsets, we have $|\Omega| = \prod_{1\leq k \leq K} |\Omega_k|$.

\begin{theorem} 
\label{thm_group}
% Assume we can cluster the queries into $K$ cluster of queries $\F_i, 1\leq i \leq K$. For each group, we can create one Sigma algebra $\Omega_k, 1\leq i \leq K,$ so that all queries in the group will be a subset of $\Omega_k.$
Suppose that there are $\q$ number of queries, with
$$
\q > \biggl( \sum_{k=1}^{K} |\Omega_k| \biggr) \cdot \sqrt{ \max_{1\leq k \leq K} |\Omega_k| }.
$$ 
Then, Algo.~\ref{algo_group} has better utility than Algorithm~\ref{algo_tradition} given the same privacy budget.
\end{theorem}

\begin{proof}
Similar to Theorem~\ref{thm_offline}, the privacy leakage and utility of the benchmark algorithm is 
$$
\v_\bud(\Alg_{1}) = \q \cdot \v, \quad
\U(\Alg_{1}) = \v^{2}.
$$

For each $k$th cluster of queries in Algo.~\ref{algo_group}, similar to the proof of Theorem~\ref{thm_offline}, 
the privacy leakage and the utility satisfy
$$ \v_{\bud}(\Alg_{\ref{algo_group}}|_{\textrm{cluster}_k}) \leq |\Omega_k| \cdot \v^{\prime},$$
$$
\U(\Alg_{\ref{algo_group}}|_{\textrm{cluster}_k}) \geq \frac{{\v^{\prime}}^{2}}{|\Omega_k|}.
$$
If we apply ${\v^{\prime}}$-DP to each sigma-algebra element. 
Thus, the overall privacy leakage and utility of Algorithm~\ref{algo_group} satisfy
\begin{align}
     &\v_\bud(\Alg_{\ref{algo_group}}) 
     \leq \sum_{1\leq k \leq K}|\Omega_k| \cdot {\v^{\prime}}  \nonumber \\ 
    & \U(\Alg_{\ref{algo_group}})
    \geq \min_{1\leq k \leq K}  \frac{{\v^{\prime}}^{2}}{|\Omega_k|} 
    = \frac{{\v^{\prime}}^{2}}{\max_{1\leq k \leq K} |\Omega_k|}.\label{eq_30}
\end{align}

Therefore, to ensure the same privacy budget, we let 
$%\begin{align}
   \v^{\prime} =   q  \cdot \v / \sum_{i=1}^K |\Omega_k| . %\label{eq_16}
$ %\end{align}
Taking the above into (\ref{eq_30}), we obtain 
$$
    \U(\Alg_{\ref{algo_group}}) \geq \v^{2} \frac{\q^2}{(\sum_{i=1}^K |\Omega_k| )^2 \cdot \max_{1\leq i \leq K }( |\Omega_k|) }.
$$
% Sigma counting has utility $\geq \v^{2} Q^2/|\Omega|^3 $.
Thus, the condition in Theorem~\ref{thm_group} guarantees $\U(\Alg_{\ref{algo_group}})  > \U(\Alg_{1})$, which completes the proof. 
\end{proof}

Theorem~\ref{thm_group} shows that the grouped sigma-counting can be more favorable than the non-grouped counterpart.
For example, suppose that we have two variables/columns, and each variable has 10 categories. So, $|\Omega_1|=|\Omega_2| = 10$, and $|\Omega|=100$. Theorem~\ref{thm_offline} requires $\q > 1000$, while Theorem~\ref{thm_group} requires $\q > 63$. 

% \vspace{-0.1in}
\subsection{Online sigma-counting}
 	\begin{algorithm}[H]
		\centering
		\caption{Online sigma-counting ($\Alg_{4}$)}\label{algo_online}
		\footnotesize
		\begin{algorithmic}[1]
			\renewcommand{\algorithmicrequire}{\textbf{Input:}}
			\renewcommand{\algorithmicensure}{\textbf{Output:}}
			\REQUIRE A set of queries $\F$, and online arrived queries $\A_1, \A_2, \ldots$
			\STATE Apply Algo.~\ref{algo_group} with $\F$ and obtain sigma-algebras $(\Omega_k, \Sigma_k)$, $k=1,\ldots,K$
			\IF{$\A \in \Omega_k$ for some $k$}
			\STATE
			find the $\Omega_k$ that contains $\A$ and apply the corresponding mechanism (\ref{sig})
			\ELSE 
			\STATE apply Algo.~\ref{algo_tradition} with budget $\v_{\bud}- \sum_{k=1}^{K} |\Omega_k| \v$
			\ENDIF
            % \FOR{$t=1 \to K \de \lfloor (\v_{\bud}- \sum_{1\leq i \leq K} |\Omega_k| \v )/\v \rfloor$}
% 			\STATE 
% 			\ENDFOR
			\ENSURE The returned queries $\Nt(\A_1),\Nt(\A_2),\ldots$
		\end{algorithmic}
		\end{algorithm}
% \end{minipage}
% \end{wrapfigure}
%
%
In practice, we often encounter streaming queries that need to be processed in a timely fashion. 
For a sequence of queries online received, we propose to use the following solution. 
We regularly update the sigma-counting algorithm in batches (say every day or week), using the techniques developed in the previous subsections. 
Within any two updates, some of the received queries will fall into the established sigma-algebras. We will apply sigma-counting for them. For the remaining queries not belonging to 
a sigma-algebra, we apply the benchmark method (Algo.~\ref{algo_tradition}).
The algorithm pseudocode is summarized in Algo.~\ref{algo_online}. It has the same time and space complexity as Algo.~\ref{algo_group}.

\begin{theorem} 
\label{thm_online} 
Suppose that the total number of queries $\q$ satisfies
 \begin{align}
\q > \frac{ \big( \sum_{k=1}^{K} |\Omega_k| \big) \sqrt{\max_{k=1}^{K} |\Omega_k| }  }{1 -p}.  \label{cond3}
\end{align}
Then, Algo.~\ref{algo_online} has better utility than Algorithm~\ref{algo_tradition} given the same privacy budget.
\end{theorem}

\begin{proof}
% Similar to Theorem~\ref{thm_offline}, the privacy leakage and utility of benchmark algorithm is $$
% \v_\bud(\Alg_{1}) = \q \cdot \v.\label{eq_12}
% \quad
% \U(\Alg_{1}) = \v^{2}.$$

In Algo.~\ref{algo_online}, each query has two possibilities: either belonging to a counted sigma-algebra $(\Omega_k, \F_k)$ (named `class-one') or not (named `class-two'). 
There are $N \cdot (1-p)$ number of class-one queries.
% Denote the sigma-algebras for class-one queries as $(\Omega_k, \Sigma_i)$. 
The proof of Theorem~\ref{thm_group} indicates that the privacy leakage and utility of these queries satisfy
\begin{align}
\v_{\bud}(\Alg_{\ref{algo_online}}|_{\texttt{class-one}}) = \sum_{k=1}^K |\Omega_k| \cdot \v^\prime, \label{eq_20} \\ 
\U(\Alg_{\ref{algo_online}}|_{\texttt{class-one}}) \geq \frac{{\v^\prime}^{2}}{\max_{1\leq k \leq K} |\Omega_k| } . \label{eq_21}
\end{align}
For the $N \cdot p$ number of class-two queries, since we apply Algo.~\ref{algo_tradition}, the privacy leakage and utility are
\begin{align}
    &\v_{\bud}(\Alg_{\ref{algo_online}}|_{\texttt{class-two}}) = \q \cdot p \cdot \v,  \label{eq_22} \\
    &\U(\Alg_{\ref{algo_online}}|_{\texttt{class-two}}) = \v^2.\label{eq_23}
\end{align}

Combining (\ref{eq_20}) and (\ref{eq_22}), the overall privacy leakage of Algo.~\ref{algo_online} is
$
\sum_{k=1}^K |\Omega_k| \cdot \v^{\prime} +  \q \cdot p \cdot \v.
$
From (\ref{eq_21}) and (\ref{eq_23}), the overall utility is at least 
$
  (1-p) \cdot {\v^{\prime}}^{2}/\max_{1\leq k \leq K} |\Omega_k| + p \cdot  \v^{2}.
$
To ensure the same privacy budget with the benchmark algorithm, we let
$
\v^{\prime} = (1-p)\cdot \q \cdot \v / \sum_{k =1}^{K} |\Omega_k|.
$
Then, the utility is at least 
% \jun{adjust the equation a bit more}
\begin{align}
    & \U(\Alg_{\ref{algo_online}}) \geq  \nonumber \\
&  \biggl( \frac{(1-p)^3 \cdot Q^2}{(\sum_{k=1}^K |\Omega_k|)^2 \cdot ( {\max_{1\leq k \leq K} |\Omega_k| } )} + p \biggr)\v^2.
\end{align}
It follows that Algo.~\ref{algo_online} is better than Algo.~\ref{algo_tradition} (whose utility is $\v^2$) under Condition~\ref{cond3}.
%This completes the proof. 
\end{proof}

Compared with Theorem~\ref{thm_group}, the condition in Theorem~\ref{thm_online} involves an extra factor of $1/(1-p)$.
It is needed to account for the $p$ percentage of queries that need to be independently randomized. 
% Similar to the previous theorem, this bound seems worse than the previous theorem, but it is better. This algorithm accommodates outlying queries which is $100\cdot p $ percent of all queries. As long as $p$ is a small number, so that equation~(\ref{cond3}) holds, this new algorithm still has better utility. 

\subsection{Evolving database}
% \vspace{-0.1in}

% \begin{wrapfigure}{R}{0.45\textwidth}
% \vspace{-8mm}
% \begin{minipage}{0.45\textwidth}
 	\begin{algorithm}[H]
		\captionof{algorithm}{Sigma-counting for evolving database $\Alg_{5}$}\label{algo_evolve}
		\footnotesize
		\begin{algorithmic}[1]
			\renewcommand{\algorithmicrequire}{\textbf{Input:}}
			\renewcommand{\algorithmicensure}{\textbf{Output:}}
			\REQUIRE %A privacy budget $\v_{\bud}$, 
% 			A set of $\q$ queries repetitively sent to the database at each time $t$, denoted by $\A_{i,t}$, $i=1\ldots,\q,$ $t=1,\ldots,T$ 
A query $\A$ repetitively sent to the database at each time $t$, denoted by $\A_{t}$, $t=1,\ldots,T$
            % \FOR{$i=1 \to \q$}
             \IF{$t=1$}
             \STATE Apply (\ref{eq_4}) to obtain $\Nt(\A_{1})$
             \ELSE 
             \STATE Let
             $ %\begin{align}
                % \Nt(\A_{i,t}) = \Nt(\A_{i,1}) + \N(\A_{i,t}) - \N(\A_{i,1}) %\label{eq_ts}
                \Nt(\A_{t}) = \Nt(\A_{1}) + \N(\A_{t}) - \N(\A_{1}) %\label{eq_ts}
            $ %\end{align}
             \ENDIF
% 			\ENDFOR
			\ENSURE The returned counts $\Nt(\A_{t})$ for each $t$
		\end{algorithmic}
		\end{algorithm}
% \end{minipage}
% \end{wrapfigure}
%
In more complex application scenarios, the database $\D$ being queried can be time-varying. For example, a cohort of individuals occasionally updates their demographic information in the database. Responding to the same query sent to a time-varying database at different time steps is more likely to leak individual information than a static database. We denote the database as $\D_{1:T} \de \{D_t, t= 1,\ldots,T\}$, where $t$ represents the time step and $D_t$ denotes the database at time $t$. We first introduce  differential privacy for a time-varying database.
We define $(\D_{1:T},\DD_{1:T})$ to be adjacent, if $\D_t$ and $\DD_t$ are the same for all $t$ except for one row/individual, which belongs to $\D_1$ but not belong to $\DD_{1:T}$.
In other words, an individual of $\D_1$ may be removed or re-added into $\D_t$ at any future time step $t$, but it will not be counted in $\DD_t$ for any $t$.

% \vspace{-0.1in}
\begin{definition}[$\v$-differential privacy]\label{eq_DP2}
A randomized mechanism $\M: \D_{1:T} \mapsto \mathcal{R}^T$ satisfies $\v$-differential privacy (DP) if for any two adjacent inputs $\D_{1:T},\DD_{1:T}$ and for any subset of outputs $S \subseteq \mathcal{R^T}$, it holds that 
$$
e^{-\v} \leq  \P(\M(\D_{1:T})\in S) / \P(\M(\DD_{1:T}) \in S)  \leq e^{\v} .
$$
\end{definition}
% \vspace{-0.1in}

Let us consider a counting query $\A$ that will be repetitively sent to the database $\D_t$ at different time $t=1,\ldots,T$.
With a slight abuse of notation, we denote $A_t$ as the query at time $t$. Let $\AA_t$ denote the same query sent to the adjacent database $\DD_t$.
In our context, a mechanism meets $\v$-DP if
\begin{align*}
e^{-\v} & \leq  \frac{\P(\Nt(\A_t) = s_t, t = 1,\ldots,T) }{ \P(\Nt(\AA_t) = s_t, t = 1,\ldots,T)  } \leq e^{\v}
\end{align*}
for any choice of integers $s_t$.

% With this definition, the perturbation philosophy remains the same: we add the constant perturbations (concerning time) to the database and prove that certain privacy can still be guaranteed.
We propose the following solution to improve the utility-privacy trade-off.
At time $t = 1,$ we respond each query with the mechanism as in Algo.~\ref{algo_tradition}. For $t \geq 2$, we add the true difference between the counts at time $t$ and $1$. Thus, we add deterministic components to the responses.
The pseudocode is summarized in Algo.~\ref{algo_evolve}.
% For brevity, we omitted the sigma-counting part and 
For brevity, we only address the queries that do not fall into an existing sigma-counting, namely the counterpart of line 5 in Algo.~\ref{algo_online}. Also, we do not explicitly specify the privacy budget. 
% The number of queries $\q$ is understood as the maximum number that is allowed by a budget Algo.~\ref{algo_online} may or may not be used in conjunction with earlier algorithms.
The following Theorem~\ref{thm_evolve} shows that the method will achieve a better utility compared with Algo.~\ref{algo_tradition}.

\begin{theorem}\label{thm_evolve}
Assume that for any query $\A$, $\{\Nt(\A)\}_{1 \leq t \leq T}$ is a Markovian process with independent increments.
Assume that the probability of an individual remaining in $\A$ at $t=1,\ldots,T$ 
is $(1 - \rho)$ for some constant $\rho \in (0,1)$. 
Then Algo.~\ref{algo_evolve} satisfies $\v_0 \de \log \big( (1-\rho ) e^\v  + \rho e^{T\v} \big)$-DP.
\end{theorem}

Note that when $T \v, \rho$ are very small, we have
\begin{align}
    \v_0
    &= \log (( 1-\rho ) e^\v  + \rho e^{T\v} )  \nonumber \\
    & \sim \log \big((1-\rho)(1+\v) + \rho (1 + T\v)\big)  \nonumber \\
    & \sim \v \big(1 + \rho (T-1)\big). \label{eq_31}
\end{align}

The proof is included in the supplementary document.
Compared with Algo.~\ref{algo_tradition}, which requires a privacy budget of $\v T$, the budget in (\ref{eq_31}) is significantly smaller for small~$\rho$.

\section{Conclusion} \label{sec_con}
In this work, we developed a method named sigma-counting to generate privacy-preserving counting queries. 
Sigma-counting uses the notion of sigma-algebra to construct randomized counts of queries to a database.
We showed that under various circumstances  the proposed approach can be significantly better in utility-privacy trade-offs than the standard method. 
It is particularly suitable for a massive number of overlapping queries to a large database, which is common in many modern database applications. 
An interesting future problem is to relax the assumptions we made for time-varying database and further improve the utility-privacy trade-offs. 
Another future problem is to emulate the proposed approach to study other utility functions.

The \textbf{appendix} contains the proof of Theorem~\ref{thm_evolve} and some experimental studies.

\appendix
 
\subsection*{Data study: Sigma-Counting maintains the ordinal structure of queries} %\label{subsec_data1}
In this data study, we use a simulated database and a publicly available database named     Adult Census Income~\citep{kohavi1996scaling}. This experiment is intended to illustrate the property of Sigma counting that it keeps the ordinal structure of queries. The simulation database consists of $10^5$ rows and $21$ columns, of which elements are binary. We select $9$ variables from the original dataset and encode the dataset into a database with $32561$ rows and $21 $ columns, of which elements are also binary. We also simulated $10^5$ queries. The simulation is generated through independent multinomial distributions with the probabilities $0.05, 0.05, 0.9$, representing the selection of `$1$', `$0$', `both $1$ and $0$', respectively.

In this experiment, to understand the effects of number of columns, we also selected $5$ columns from the real dataset and generate a smaller encoded dataset with $11$ columns and $32561$ rows, of which elements are binary. To match this newly generated dataset, we also generate a simulation dataset with $11$ columns and $10^5$ rows. In summary, we use two smaller datasets with $11$ columns (one simulation dataset and one real dataset) and two similar larger datasets with $21$ columns.

Recall that the proposed Algorithm~\ref{algo_offline} is monotonic in the sense that for any two queries $\A_1,\A_2$ satisfying $\A_1 \subseteq \A_2$, we have $\Nt(\A_1) \leq \Nt(\A_2).$ We generate a sequence of nested query pairs, in which one query is a subset of the other. As proved in Theorem~\ref{thm_mono}, the sigma-counting method keeps the ordinal structure. 

We apply the benchmark method Algorithm~\ref{algo_tradition} and sigma-counting method Algorithm~\ref{algo_offline}, with the identical overall budget. Table~\ref{table_mono} shows that benchmark algorithm breaks ordinal structure frequently. The percentage of total order violations is worsened as the budget becomes smaller, since more noise needs to be injected. Note that the sigma-counting strictly preserves the ordinal structure, so the percentage of total order violations is precisely $0$. Moreover, by comparing the results of datasets with $11$ columns and $21$ columns, we observe that there are more total order violations on the dataset with $21$ columns. It is intuitive since the queries to a dataset with more columns tend to have larger sizes and thus more filtering, which, in turn, leads to smaller counting results. %and thus corresponding counting is generally smaller. 

% \begin{table}[tb]
% \caption{Summary of the Total Order Violations of Algorithm~\ref{algo_tradition}  
% %Queries are independent in all columns and for each column the probabilities of selecting $1$, $0$ and both $1$ and $0$ is 0.05, 0.05, 0.9.
% for different total budgets ($100, 10, 1$) and four databases.
% The sigma-counting method produces zero violation and is not reported.  
% }
% \label{table_mono}
% \begin{center}
% \begin{tabular}{cccccccc}
% \hline
% \multirow{3}{*}{Simulation Data}
% & Budget & $100$     & $10$    & $1$       \\ \cline{2-5} 
% & 11 columns & $0.01\%$  & $4.0\%$ & $26.6\%$   \\ \cline{2-5} 
% & 21 columns & $0.60\%$  & $13.2\%$ & $38.9\%$  \\ 
% \hline
% \multirow{3}{*}{\begin{tabular}[c]{@{}c@{}} Real Data\end{tabular}}
% & Budget & $100$ & $10$  & $1$   \\ \cline{2-5} 
% & 11 columns & $6.3\%$  & $14.2\%$ & $25.3\%$  \\ \cline{2-5} 
% & 21 columns & $10.3\%$  & $21.4\%$ & $32.3\%$   \\ \hline
% \end{tabular}
% \end{center}
% \end{table}

\subsection*{Data study: Sigma-Counting utility} %\label{subsec_data1}   

In this data study, we compare the querying utility of sigma-counting Algorithm~\ref{algo_online} and benchmark method Algorithm~\ref{algo_tradition}, under the same total budget. As described in Algorithm~\ref{algo_online}, the queries are partially online and the calculation of utility is based on both offline and online queries. 
In accordance with the setting of Theorem~\ref{thm_group}, we apply the following rule of clustering. For a pre-determined positive integer $u$, we treat all the subsets with size no larger than $u$ as clusters, and categorize the queries that fall into those clusters.

We will refer to $u$ as the partition size threshold. Meanwhile, the remaining queries will be addressed by Algorithm~\ref{algo_tradition}. The larger $u$, the larger proportion of queries can be categorized into the pre-determined clusters. Consequently, the overall estimation error (as quantified by the mean squared loss) is a weighted average of the errors contributed by the clustered queries and those contributed by the individually perturbed queries. 

In this experiment, we continue using the same data as in the previous section, including two smaller datasets with $11$ columns and two larger datasets with $21 $ columns.  We also simulate two sets of queries in this experiment with the same distribution as before, with sizes $10^5$ and $10^6$.
For better presentation, we define the utility of the sigma-counting algorithm as %$\Alg_{\texttt{sigma-counting}} / \Alg_{\texttt{benchmark}}$. 
$\Alg_{4} / \Alg_{1}$, namely the relative utility. 
% \begin{equation}
% \label{utility_def}
%     \frac{1/E_{sigma}^2}{1/E_{Benchmark}^2},
% \end{equation}
% where $E_{Benchmark}^2$ is the mean squared errors of Benchmark algorithm and $E_{sigma}^2$ is the mean squared errors of  Algorithm~\ref{algo_online}. \jun{Ask Ding which package he deleted}
The numerical results are summarized in Tables~\ref{table_thm4} and~\ref{table_thm4_2}. 

The experimental results indicate the following three points. 
First, Tables~\ref{table_thm4} and~\ref{table_thm4_2} show that the utility of sigma-counting is much greater than that of benchmark algorithm. In other words, the sigma-counting allows privacy-compliant queries with less error.
Second, comparing the two sets of queries with sizes $10^5$ and $10^6$, we can see that the utility improves as the number of queries increases. This is consistent with our theory that  the sigma-counting works particularly well for a large number of overlapping queries. 
Third, as the number of database columns becomes smaller, the sigma-counting performs better compared with the benchmark. This is also consistent with our intuition, since queries are more likely to be clustered for a database with a smaller number of columns. 

% It is important to note that in practice, the queries may not be uniformly distributed and it is possible to achieve higher utility. This experiment uses simulated queries which have a uniform distribution across all columns. Clustering queries is thus across a sparse space which hurts utility.

\begin{table}[tb]
\caption{Summary of the Total Order Violations of Algorithm~\ref{algo_tradition}  
%Queries are independent in all columns and for each column the probabilities of selecting $1$, $0$ and both $1$ and $0$ is 0.05, 0.05, 0.9.
for different total budgets ($100, 10, 1$) and four databases.
The sigma-counting method produces zero violation and is not reported.  
}
\label{table_mono}
\begin{center}
\begin{tabular}{cccccccc}
\hline
\multirow{3}{*}{Simulation Data}
& Budget & $100$     & $10$    & $1$       \\ \cline{2-5} 
& 11 columns & $0.01\%$  & $4.0\%$ & $26.6\%$   \\ \cline{2-5} 
& 21 columns & $0.60\%$  & $13.2\%$ & $38.9\%$  \\ 
\hline
\multirow{3}{*}{\begin{tabular}[c]{@{}c@{}} Real Data\end{tabular}}
& Budget & $100$ & $10$  & $1$   \\ \cline{2-5} 
& 11 columns & $6.3\%$  & $14.2\%$ & $25.3\%$  \\ \cline{2-5} 
& 21 columns & $10.3\%$  & $21.4\%$ & $32.3\%$   \\ \hline
\end{tabular}
\end{center}
\end{table}

\begin{table}[tb]
\caption{Summary of the sigma-counting utility, defined as the ratio of mean squared errors of the benchmark method and sigma-counting method (on the datasets with $11$ columns). 
% Queries are independent in all columns and for each column the probabilities of selecting $1$, $0$ and both $1$ and $0$ is 0.05, 0.05, 0.9.   
% The percentage of queries to which sigma-counting is applied is: $98.7\%$,$93.0\%$,$74.6\%$.
}
\label{table_thm4}
\begin{center}
\begin{tabular}{cccccccc}
\hline
\multirow{3}{*}{Simulation Data}
& Partition Size Threshold $u$
 & $3$    & $2$  & $1$
\\ \cline{2-5}
% & Percentage of queries using Sigma Counting  & $99.9\%$  & $98.7\%$  & $92.9\%$  & $73.5\%$   \\ \cline{2-6} 
& Utility, $10^5 $ Queries   & $77.1$  & $14.0$  & $3.8$ 
  \\ \cline{2-5} 
% & Percentage of queries using Sigma Counting  & $99.8\%$  & $98.7\%$  & $93.0\%$  & $73.6\%$   \\ \cline{2-6} 
& Utility, $10^6 $ Queries    & $77.5$  & $14.2$  & $3.8$ 
   \\ \cline{2-5} 
\hline
\multirow{3}{*}{\begin{tabular}[c]{@{}c@{}} Real Data \end{tabular}} 
% & Percentage of queries using Sigma Counting  & $99.8\%$  & $98.8\%$  & $93.0\%$  & $73.7\%$   \\ \cline{2-6} 
& Utility, $10^5 $ Queries   & $77.3$  & $14.2$  & $3.8$  \\ \cline{2-5} 
% & Percentage of queries using Sigma Counting  & $99.8\%$  & $98.7\%$  & $93.0\%$  & $73.6\%$   \\ \cline{2-6} 
& Utility, $10^6 $ Queries   & $76.9$  & $14.2$  & $3.8$   \\ \cline{2-5} 
\hline
\end{tabular}
\end{center}
\end{table}

\begin{table}[tb]
\caption{
Summary of the sigma-counting utility, defined as the ratio of mean squared errors of the benchmark method and sigma-counting method (on the datasets with $21$ columns). 
% Utilities are defined as ratio of mean squared errors of benchmark algorithm and mean squared errors of sigma-counting. The percentage of queries to which sigma-counting is applied is: $84.8\%$,$64.8\%$,$36.5\%$. Queries are independent in all columns and for each column the probabilities of selecting $1$, $0$ and both $1$ and $0$ is 0.05, 0.05, 0.9. 
}
\label{table_thm4_2}
\begin{center}
\begin{tabular}{cccccccc}
\hline
\multirow{3}{*}{Simulation Data}
& Partition Size Threshold $u$
& $3$    & $2$  & $1$
\\ \cline{2-5}
% & Percentage of queries using Sigma Counting  & $94.7\%$  & $84.8\%$ & $64.9\%$  & $36.5\%$  \\ \cline{2-6} 
& Utility, $10^5$ % & $0.2$  
& $5.2$ & $2.8$ & $1.6$  \\ \cline{2-5} 
% & Percentage of queries using Sigma Counting  & $94.7\%$  & $84.8\%$ & $64.8\%$  & $36.5\%$  \\ \cline{2-6} 
& Utility, $10^6 $ Queries    & $6.5$ & $2.8$ & $1.6$    \\ \cline{2-5} 
\hline
\multirow{3}{*}{\begin{tabular}[c]{@{}c@{}} Real Data \end{tabular}} 
% & Percentage of queries using Sigma Counting  & $94.8\%$  & $84.7\%$ & $64.7\%$  & $36.5\%$  \\ \cline{2-6} 
& Utility, $10^5 $ Queries   & $5.1$ & $2.8$ & $1.6$ \\ \cline{2-5} 
% & Percentage of queries using Sigma Counting  & $94.8\%$  & $84.8\%$ & $64.8\%$  & $36.4\%$  \\ \cline{2-6} 
& Utility, $10^6 $ Queries   & $6.6$ & $2.8$ & $1.6$  \\ \cline{2-5} 
\hline
\end{tabular}
\end{center}
\end{table}

\subsection*{Proof of Theorem~5}
\begin{proof}
Recall that $D_{1:T}$ is an evolving database. Without loss of generality, for any individual Alice who is in the query $A$ at time $t=1$, we denote $\DD_{1:T}$ as the adjacent evolving database without Alice. Recall that $A$ is a counting query, and we denote $A_t$ as the query at time $t$ and $\AA_t$ denote the same query sent to the adjacent database $\DD_t$.

For notational convenience, we let
$$
\Nt_t = \Nt(A_t), \quad
\Nt_t^\prime = \Nt(\AA_t) .
$$
Denote the following event as  event $Z$: the individual Alice always remains in the query $A$ from $t=1$ to $T$.
Note that we have $ \hat{N}^\prime_1 = \hat{N}_1 - 1, $ due to the fact that Alice is in $A_1$ at time $t=1$ in database $\D_1$ but not in $\DD_1.$
Therefore we have
$$
\frac{\P(\Nt_1 = s_1)}{\P(\Nt^\prime_1 = s_1)} = \frac{\P(\Nt^\prime_1 = s_1-1)}{\P(\Nt^\prime_1 = s_1)}  \leq e^{\v}.
$$
% And notice if $Z$ does not opt out of the group. 
Since $N^\prime_t$ is a Markovian process, and we denote the event: 
$$
E_{t, t-1} = \{ \Nt_t - \Nt_{t-1} = s_t - s_{t-1} \},
$$
$$
E^\prime_{t, t-1} = \{ \Nt_t^\prime - \Nt_{t-1}^\prime = s_t - s_{t-1} \}.
$$
Therefore we have
\begin{align}
&\frac{\P(\Nt_1 = s_1, \cdots, \Nt_T = s_T)}{\P(\Nt^\prime_1 = s_1, \cdots, \Nt^\prime_T = s_T)} \nonumber \\ 
&=
\frac{\P(\Nt_1 = s_1, E_{t, t-1}, 2\leq t \leq T, Z)  }
{\P(\Nt^\prime_1 = s_1,  E^\prime_{t, t-1} , 2\leq t \leq T) } + \nonumber \\ 
&\quad \frac{\P(\Nt_1 = s_1, E_{t, t-1} ,2\leq t \leq T, Z^c)  }
{\P(\Nt^\prime_1 = s_1,   E^\prime_{t, t-1},2\leq t \leq T ) } \nonumber \\ 
&= I + II. \nonumber %\label{time_series_estimate}
\end{align}
Note that in the definition of $\Nt^\prime_t$, the individual Alice is not included from the database $\DD_{1:T}$. Thus, $Z$ and $Z^c$ are independent of $\{\Nt^\prime_t\}_{1\leq t \leq T}$.
Note that $P(Z) = 1 - \rho$ and $\{\Nt^\prime_i, 1\leq t \leq T\}$ has independent increments.
For term $I$, notice, Alice is always in the database, therefore we have:
\begin{align*}
    \Nt^\prime_1 = \Nt_1 - 1 = s_1-1, \Nt^\prime_t - \Nt^\prime_1 = \Nt_t - \Nt_1,
\end{align*}
We denote 
$$
E^\prime_{t, 1} = \{ \Nt_t^\prime - \Nt_1^\prime = s_t - s_1 \}
$$
Therefore, term $I$ can be written as:
\begin{align*}
 I & = \frac{\P(\Nt^\prime_1 = s_1-1, E^\prime_{t, 1}, 2\leq t \leq T) ( 1- \rho)}
{\P(\Nt^\prime_1 = s_1,   E^\prime_{t, 1}, 2\leq t \leq T ) }  \nonumber \\
& = ( 1-\rho) \frac{ \P(\Nt^\prime_1 = s_1-1) \prod_{t=2}^T  \P(E^\prime_{t, t-1}) }{\P(\Nt^\prime_1 = s_1) \prod_{t=2}^T \P(E^\prime_{t, t-1})} \nonumber \\
& = ( 1-\rho) \frac{ \P(\Nt^\prime_1 = s_1-1)  }{\P(\Nt^\prime_1 = s_1) } \leq (1-\rho) e^{\v}
\end{align*}
For term $II$, since
$$\hat{N}^\prime_t = \hat{N}_t \textrm{ or } \hat{N}_t - 1 $$  depending on whether Alice is in the query at time $t$ ($2 \leq t \leq T$),
we have
$$\Nt_t^\prime - \Nt_{t-1}^\prime = s_{t} - s_{t-1} \textrm{ or } s_t - s_{t-1}  -1.
$$
Similar to the calculation for term $I$, and recall definition of $E^\prime_{t, t-1}$.
$$
E^\prime_{t, t-1,1} = \{ \Nt_t^\prime - \Nt_{t-1}^\prime = s_t - s_{t-1} \}.
$$
Therefore we have 
\begin{align*} 
% II = &  \rho \frac{\P(\Nt^\prime_1 = s_1-1, \Nt^\prime_t - \Nt^\prime_{t-1} = s_t - s_{t-1} \textrm{ or } s_t - s_{t-1} -1,2 \leq t \leq T)}{\P(N^\prime_1 = s_1, \Nt^\prime_2 - \Nt^\prime_1 = s_t - s_{t-1}, 2 \leq t \leq T)} \nonumber \\ 
II = &  \rho \frac{\P(\Nt^\prime_1 = s_1-1, E^\prime_{t, t-1} \textrm{ or } E^\prime_{t, t-1,1},2 \leq t \leq T)}{\P(N^\prime_1 = s_1, E^\prime_{t, t-1}, 2 \leq t \leq T)} \nonumber \\ 
=
& \rho \frac{\P(\Nt_1 = s_1 -1) }{\P(\Nt^\prime_1 = s_1)} \cdot  \prod_{t=2}^T R_t,  \nonumber \\
\end{align*}
where 
\begin{align}
R_t & = \frac{   \P(\Nt^\prime_t -\Nt^\prime_{t-1}  = s_t - s_{t-1} \textrm{ or } s_t - s_{t-1} - 1)}{ \P(\Nt^\prime_t - \Nt^\prime_{t-1} = s_t - s_{t-1}  )} \nonumber \\
& = \frac{\P(E^\prime_{t, t-1} \textrm{ or } E^\prime_{t, t-1,1})}{\P(E^\prime_{t, t-1} )}.
\end{align}

Notice  we add perturbation to each increment $\Nt^\prime_t - \Nt^\prime_{t-1}$. 
Depending on whether Alice remains in the query $A$ at time $t$,  $R_t$ is either 1 or
\begin{align*}
& R_t = \frac{ \P(\Nt^\prime_t - \Nt^\prime_{t-1} = s_t - s_{t-1} - 1 )}{ \P(\Nt^\prime_t - \Nt^\prime_{t-1} = s_t - s_{t-1})} \leq e^{\v}.
% \nonumber \\
% = 
% & \frac{ \sum_k P(\Nt'_2 - k = S_2 - S_1 - 1, \Nt_1 = k )}{ \sum_k P(\Nt'_2 - k = S_2 - S_1, \Nt_1 = k)} \nonumber \\
% & \leq e^{\v}
\end{align*}
Therefore, we have
$$II \leq (1-\rho) e^{T \v},$$
and further 
$
I+II \leq 
e^{\v_0}, $
where $$\v_0= \log \big( (1-\rho ) e^\v  + \rho e^{T\v} \big).$$
Similarly, it can be verified that 
$$
I+II \geq 
e^{-\v_0'} \de ( 1 - \rho) \cdot e^{-\v} + \rho e^{- T \v}.
$$
Since
$$
\big( ( 1 - \rho) \cdot e^{\v} + \rho e^{ T \v} \big)
\cdot
\big( ( 1 - \rho) \cdot e^{-\v} + \rho e^{- T \v} \big)
> 1,
$$
we have
$e^{\v_0 -\v_0' } > 1$, or equivalently,  $\v_0 > \v_0'.$ 
Consequently, Algorithm~5 %\ref{algo_evolve}
satisfies $\v_0$-DP.
\end{proof}

\balance
\bibliography{privacy,J}
% \bibliographystyle{IEEEtran}
% \setcitestyle{authoryear,open={((},close={))}} %Citation-related commands
% \bibliographystyle{abbrvnat}
% \bibliographystyle{ieeetr}
\bibliographystyle{plainnat}

\end{document}